\documentclass[12pt]{l4dc2021}

\newcommand{\reffig}[1]{Figure \ref{fig:#1}}
\newcommand{\refeq}[1]{(\ref{eq:#1})}
\newcommand{\refdef}[1]{Definition \ref{def:#1}}
\newcommand{\reftheo}[1]{Theorem \ref{theo:#1}}
\newcommand{\refprop}[1]{Proposition \ref{prop:#1}}

\newcommand{\refalgo}[1]{Algorithm \ref{algo:#1}}
\newcommand{\refrema}[1]{Remark \ref{rema:#1}}
\newcommand{\refsec}[1]{Section \ref{sec:#1}}

\newenvironment{psmallmatrix}
  {\left(\begin{smallmatrix}}
  {\end{smallmatrix}\right)}

\usepackage{tikz}
\usepackage{pgfplots}
\usepackage{standalone}
\usetikzlibrary{shapes,arrows}

\tikzstyle{block} = [draw, fill=blue!20, rectangle, 
    minimum height=3em, minimum width=6em]
\tikzstyle{sum} = [draw, fill=blue!20, circle, node distance=1cm]
\tikzstyle{input} = [coordinate]
\tikzstyle{output} = [coordinate]
\tikzstyle{pinstyle} = [pin edge={to-,thin,black}]

\title[Data-Driven Controller Design]{Data-Driven Controller Design via Finite-Horizon Dissipativity}
\usepackage{times}

\author{\Name{Nils Wieler} \Email{nils.wieler@ist.uni-stuttgart.de}\\
\Name{Julian Berberich} \Email{julian.berberich@ist.uni-stuttgart.de}\\
\Name{Anne Koch} \Email{anne.koch@ist.uni-stuttgart.de}\\
\Name{Frank Allgöwer} \Email{frank.allgower@ist.uni-stuttgart.de}\\
 \addr Institute for Systems Theory and Automatic Control, University of Stuttgart, 70550 Stuttgart, Germany}

\begin{document}

\maketitle

\begin{abstract}%
Given one open-loop measured trajectory of a single-input single-output discrete-time linear time-invariant system, we present a framework for data-driven controller design for closed-loop finite-horizon dissipativity. First, we parametrize all closed-loop trajectories using the given data of the plant and a model of the controller. We then provide an approach to validate the controller by verifying closed-loop dissipativity in the standard feedback loop based on this parametrization. We use these conditions to design controllers leading to closed-loop dissipativity based on a quadratic matrix inequality feasibility problem. Finally, the results are illustrated with a simulation example.
\end{abstract}


\section{Introduction}
The majority of the control theory literature is based on knowledge of a mathematical model for the system one wants to control. This model-based theory is well-established and successfully used in many applications, where modeling by first principles is possible. In the last years, the complexity of systems has increased significantly and with the increasing amount of available data, data-driven methods have become more and more important. In many data-driven methods, the modeling part is skipped and the given data is used directly for controller design. However, many of the developed data-driven methods lack the well-known guarantees for stability, performance and robustness from the model-based setting. A summary of existing data-driven control methods is available in \cite{hou2013}. The so-called Fundamental Lemma by \cite{willems2005} has gained increasing attention in data-driven control since it provides a powerful approach to replace the usual mathematical model with a representation of the system directly on the basis of data. It was successfully applied in several fields, for example in data-driven simulation (\cite{markovsky2008}), system analysis (e.g. \cite{maupong2017,romer2019,koch2020}), controller design (e.g. \cite{markovsky2007,maupong2017behav,persis2020,berberich2020State}) and model predictive control (\cite{coulson2019,berberich2020MPC}). In this paper, we investigate data-driven controller synthesis for closed-loop dissipativity, using only measured input-output data, whereas most existing works with a similar goal require availability of state measurements. To this end, we first present a parametrization of all closed-loop trajectories in the standard feedback loop, given open-loop data of an unknown single-input single-output (SISO) plant and a model of the controller, employing Willems Fundamental Lemma. We then use this parametrization to provide necessary and sufficient conditions for closed-loop dissipativity. Finally, we show how these conditions can be used for controller design by translating them into a quadratic matrix inequality (QMI) feasibility problem, which can be solved using, e.g., difference of convex functions (DC) programming methods (\cite{dinh1997,thi2018}). We conclude the paper with an application to an example system.

\section{Notation} \label{sec:notation}
The $n\times n$ identity matrix is denoted by $I_{n}$ and the $n \times m$ zero matrix is denoted by $0_{n \times m}$ . We use $A^{\perp}$ to denote a basis matrix of the kernel of the matrix $A$ and the column space of A is denoted by $\mathrm{col}(A)$. The symbol $\otimes$ is used for the Kronecker product. The set of real symmetric matrices of dimension $n \times n$ are denoted by $\mathbb{S}^{n}$, where the superscript is omitted if the dimensions are obvious.
 For a finite sequence $\{x_{k}\}_{k=0}^{N-1}$ of length $N$, $x$ is used to denote either the sequence itself or the stacked vector containing the components of the sequence. The Hankel matrix $H_{L}(x)$ of depth L and the lower-triangular Toeplitz matrix $T_{L}(x)$ for such a sequence $x$ are given by
\begin{equation*}
	H_{L}(x) = \begin{pmatrix}
		{x}_{0} && {x}_{1} && \cdots && {x}_{N-L} \\
		{x}_{1} && {x}_{2} && \cdots && {x}_{N-L+1}\\
		\vdots && \vdots && \ddots && \vdots \\
		{x}_{L-1}&&{x}_{L}&& \cdots && {x}_{N-1}
	\end{pmatrix}, \quad T_{L} (x)= 
	\begin{pmatrix}
		x_{0} & 0 & \cdots & 0 \\
		x_{1} & x_{0}&\ddots & \vdots\\
		\vdots & \vdots & \ddots& 0 \\
		x_{L-1} & x_{L-2} &\cdots& x_{0}
	\end{pmatrix}.
\end{equation*}

\section{Preliminaries} \label{sec:preliminaries}
First, in \refsec{tracjectory_based_representation_of_lti_systems}, we introduce the main result of \cite{willems2005}, which can be used to parametrize all input-output trajectories of a discrete-time (DT) linear time-invariant (LTI) system using only measured data. It is explained in detail for state-space systems in \cite{berberich2020} and proven in \cite{vanWaarde2020}. In \cite{romer2019}, a data-driven formulation of dissipativity using this parametrization of trajectories is presented, which builds the foundation for the results in this paper and which is outlined in \refsec{data_based_dissipativity_characterization}.

\subsection{Trajectory-based representation of LTI systems} \label{sec:tracjectory_based_representation_of_lti_systems}
The following definition gives a precise meaning to a trajectory of an LTI system.
\begin{definition}
	An input-output sequence $\{u_{k},y_{k}\}_{k=0}^{N-1}$ with $u_{k} \in \mathbb{R}^{m}$, $y_{k} \in \mathbb{R}^{p}$ is a trajectory of a DT LTI system $G$, if there exists an initial condition $\hat{x} \in \mathbb{R}^{n}$ and a state sequence $\{x_{k}\}_{k=0}^{N}$ such that
	\begin{equation*}
			x_{k+1} = A x_{k} + Bu_{k}, \quad x_{0} = \hat{x}, \quad y_{k} =  Cx_{k} + Du_{k}
	\end{equation*}
	for $k=0,1,\hdots,N-1$, where (A,B,C,D) is a minimal realization of G.
	\label{def:trajectory_of_lti}
\end{definition}
An essential assumption for the following results is that the input sequence of the measured trajectory is persistently exciting, as characterized via following definition.
\begin{definition}
	A sequence $\{c_{k}\}_{k=0}^{N-1}$, $c_{k} \in \mathbb{R}^{q}$, is persistently exciting of order L, if $\text{rank}(H_{L}(c)){=}qL$.
\label{def:persistence_of_excitation}
\end{definition}
The following result, originally proven in \cite{willems2005} and reformulated in \cite{berberich2020}, provides a parametrization of all trajectories of an unknown system using only one measured data trajectory.
\begin{theorem}[{\cite{willems2005}}]
	Suppose $\{ u_{k},y_{k}\}_{k=0}^{N-1}$ is a trajectory of an LTI system $G$, where $u$ is persistently exciting of order $L+n$. Then, $\{ \bar{u}_{k},\bar{y}_{k}\}_{k=0}^{L-1}$ is a trajectory of $G$ if and only if there exists $\alpha \in \mathbb{R}^{N-L-1}$             
	such that
	\begin{equation}
		H_{L}(u,y) \alpha = \begin{pmatrix}\bar{u} \\ \bar{y} \end{pmatrix}, \quad \text{where} \quad H_{L}(u,y) = \begin{pmatrix}H_{L}(u) \\ H_{L}(y)\end{pmatrix}.
	\label{eq:foundation}
	\end{equation}
	\label{theo:data_based_trajectory_representation}
\end{theorem}
Based on one persistently exciting trajectory, it is possible to build any other trajectory of that system by taking linear combinations of the measured trajectory and time-shifts thereof. In other words, the stacked Hankel matrix in \refeq{foundation} spans the whole space of trajectories of length $L$ of the system $G$.

\subsection{Data-based dissipativity characterization} \label{sec:data_based_dissipativity_characterization}
In this section, we state the results of \cite{romer2019}, where \reftheo{data_based_trajectory_representation} is used to verify dissipativity properties of general multi-input multi-output DT LTI systems. The following definition of dissipativity from \cite{hill1980} is a common definition in the input-output context.
\begin{definition}
	The system $G$ is dissipative w.r.t. the supply rate $\Pi = \begin{pmatrix}Q & S \\ S^{\top} & R\end{pmatrix}$ with $Q \in \mathbb{S}^{m}$, $R \in \mathbb{S}^{p}$ and $S \in \mathbb{R}^{m \times p}$, if
	\begin{equation}
			\sum_{k=0}^{r} \begin{pmatrix} u_{k} \\ y_{k}\end{pmatrix}^{\top} \Pi \begin{pmatrix} u_{k} \\ y_{k}\end{pmatrix} \geq 0, \quad \forall r \geq 0
	\label{eq:input_output_dissipativity}
	\end{equation}
	for all trajectories $\{u_{k}, y_{k}\}_{k=0}^{\infty}$ of $G$ with initial condition $x_{0}=0$, where $x$ is the state of an arbitrary minimal realization of $G$.
	\label{def:dissipativity}
\end{definition}
Dissipation inequalities with quadratic supply rates represent for example the $L_{2}$-gain of a system or passivity. This input-output definition is equivalent to the well-known dissipativity definition by (\cite{willems1972i,willems1972ii}) for controllable LTI systems (compare proofs in \cite{hill1980}), and is therefore a natural starting point to infer dissipativity properties using only input-output trajectories. Since we only consider finite input-output data trajectories, we use a relaxed version of \refdef{dissipativity}, the finite-horizon dissipativity characterization called $L$-dissipativity as introduced in \cite{maupong2017}. In more detail, we consider inequality \refeq{input_output_dissipativity} only for $r < L$ (cf. \cite{romer2019}, Definition 4), which can in the LTI-case be equivalently formulated via the following result.
\begin{proposition}[{\cite{romer2019}, Proposition 1}]
	The system $G$ is L-dissipative w.r.t. the supply rate $\Pi$ if and only if
	\begin{equation}
			\sum_{k=0}^{L-1} \begin{pmatrix} u_{k} \\ y_{k}\end{pmatrix}^{\top} \Pi \begin{pmatrix} u_{k} \\ y_{k}\end{pmatrix} \geq 0
	\label{eq:L-dissipativity_cond}
	\end{equation}
	for all trajectories $\{u_{k}, y_{k}\}_{k=0}^{L-1}$ of $G$ with initial condition $x_{0}=0$, where $x$ is the state of an arbitrary minimal realization of $G$.
	\label{prop:L-dissipativity}
\end{proposition}
Under weak assumptions, taking $L \to \infty$, $L$-dissipativity is equivalent to infinite-horizon dissipativity (see \cite{koch2020}). While the well-known interconnection guarantees do generally not hold without taking $L \to \infty$, good results can be achieved in practice when taking $L$ large enough.

In the following, the main result of \cite{romer2019} is stated. It combines \reftheo{data_based_trajectory_representation} and \refprop{L-dissipativity}, which results in a data-based characterization of $L$-dissipativity. More precisely, the trajectories in \refeq{L-dissipativity_cond} are parametrized using the stacked Hankel matrices in \refeq{foundation}. To handle the fact that \refeq{L-dissipativity_cond} has to hold only for trajectories with zero initial conditions, we define the matrix
\begin{equation}
	\tilde{V}^{\nu}_{L} = 
	\begin{pmatrix}
		I_{m\nu} & 0_{m\nu \times m(L-\nu)} & 0_{m\nu \times p\nu} & 0_{m\nu \times p(L-\nu)} \\
		0_{p\nu \times m\nu} & 0_{p\nu \times m(L-\nu)} & I_{p\nu} & 0_{p\nu \times p(L-\nu)}
	\end{pmatrix}
\label{eq:V_L_nu}
\end{equation}
for some positive integer $\nu \leq L$. This implies that for any trajectory $\{u_{k},y_{k}\}_{k=0}^{L-1}$ of length $L$, $\tilde{V}^{\nu}_{L}\begin{pmatrix}u^{\top} & y^{\top} \end{pmatrix}^{\top} = 0$ if and only if $u_{0} = \hdots = u_{\nu-1} = 0$ and $y_{0} = \hdots = y_{\nu-1} = 0$. Therefore, the space of all trajectories of length $L$ with the first $\nu$ entries equal to zero is equal to the image of $H_L(u,y)V_L^\nu(u,y)$, where  $V_L^\nu(u,y)=(\tilde{V}^\nu_LH_L(u,y))^\perp$. For rewriting the sum in \refeq{L-dissipativity_cond} as a vector-matrix product, we define $\Pi_{L} = \begin{pmatrix}I_{L} \otimes Q & I_{L} \otimes S \\ I_{L} \otimes S^{\top} & I_{L} \otimes R \end{pmatrix}=\begin{pmatrix}	Q_{L} & S_{L} \\	S_{L}^{\top} & R_{L}\end{pmatrix}$. Using these definitions, the following theorem provides a data-based characterization of finite-horizon dissipativity.
\begin{theorem}[{\cite{romer2019},Theorem 2}]
	Suppose $\{u_{k},y_{k}\}_{k=0}^{N-1}$ is a trajectory of an LTI system $G$, where u is persistently exciting of order $L+n$. Then, for every $\nu$ with $n \leq \nu < L$, the system $G$ is $(L-\nu)$-dissipative w.r.t. the supply rate $\Pi$ if and only if 
		\begin{equation}
			V_{L}^{\nu}(u,y)^{\top} H_{L}(u,y)^{\top} \Pi_{L} H_{L}(u,y) V_{L}^{\nu}(u,y) \succcurlyeq 0.
		\label{eq:L-dissipative_data}
		\end{equation}
	\label{theo:data_based_L-dissipativity}
\end{theorem}
\reftheo{data_based_L-dissipativity} provides a simple data-based definiteness condition to guarantee $L$-dissipativity. While the result remains true if $\nu \geq n$ is replaced by $\nu \geq \underline{l}$ for the lag $\underline{l}$ of $G$, we use $\nu \geq n$ since $\underline{l}=n$ for observable SISO systems. In the following sections, we use \refeq{L-dissipative_data} to verify closed-loop $L$-dissipativity for a given controller in the standard feedback loop and show how to design a controller for a desired closed-loop dissipativity specification.

\section{Data-driven controller validation for closed-loop dissipativity} \label{sec:data_driven_controller_validation_for_closed_loop_dissipativity}
This section deals with the problem of verifying closed-loop $L$-dissipativity in the standard feedback loop (see \reffig{closed_loop}) for a given controller $K$,
\begin{figure}[h]
	\centering
	\begin{minipage}{0.48\textwidth}
\centering
	\begin{tikzpicture}
\draw[thick, fill=white!10] (0,0.1) rectangle (1.5,0.9);
\draw[thick, fill=white!10] (2.5,0.1) rectangle (4.0,0.9);
\draw[thick,->] (1.5,0.5) -- (2.5,0.5);
\draw[thick,->] (4,0.5) -- (5.5,0.5);
\node[anchor=south east] at (2.25,0.5) {\small $u$};
\node[anchor=south east] at (-1.25,0.5) {\small $r$};
\node[anchor=south east] at (-0.3,0.5) {\small $e$};
\node[anchor=south east] at (5.5,0.5) {\small $z$};

\draw[thick] (-1.0,0.5) circle(0.1);
\node[anchor=north east] at (-1,0.5) {\small $-$};

\draw[thick,->] (-1.75,.5) -- (-1.1,0.5);
\draw[thick,->] (-1,-.5) -- (-1,0.4);
\draw[thick,-] (4.75,-0.5) -- (4.75,0.5);
\draw[thick,->] (-0.9,0.5) -- (0,0.5);
\draw[thick,-] (-1,-.5) -- (4.75,-.5);

\node at (3.25,0.5) {\small $G$};
\node at (0.75,0.5) {\small $K$};
\end{tikzpicture}
	\caption{Standard feedback loop.}
	\label{fig:closed_loop}
\end{minipage}
\begin{minipage}{0.48\textwidth}
	\centering
	\begin{tikzpicture}
\draw[thick, fill=white!10] (0,0.1) rectangle (1.5,0.9);
\draw[thick, fill=white!10] (2.5,0.1) rectangle (4.0,0.9);
\draw[thick,->] (1.5,0.5) -- (2.5,0.5);
\draw[thick,->] (4,0.5) -- (5.5,0.5);
\node[anchor=south east] at (2.25,0.5) {\small $w$};
\node[anchor=south east] at (-1.25,0.5) {\small $r$};
\node[anchor=south east] at (-0.3,0.5) {\small $e$};
\node[anchor=south east] at (5.5,0.5) {\small $z$};

\draw[thick] (-1.0,0.5) circle(0.1);
\node[anchor=north east] at (-1,0.5) {\small $-$};

\draw[thick,->] (-1.75,.5) -- (-1.1,0.5);
\draw[thick,->] (-1,-.5) -- (-1,0.4);
\draw[thick,-] (4.75,-0.5) -- (4.75,0.5);
\draw[thick,->] (-0.9,0.5) -- (0,0.5);
\draw[thick,-] (-1,-.5) -- (4.75,-.5);

\node at (3.25,0.5) {\small $K$};
\node at (0.75,0.5) {\small $G$};
\end{tikzpicture}
	\caption{Interchanged standard feedback loop.}
	\label{fig:closed_loop_interchanged}
\end{minipage}
\end{figure}
without knowledge of a mathematical model of the plant $G$. It is assumed that the controller $K$ as well as the plant $G$ are both SISO LTI systems. Furthermore, a model of the controller is given via its impulse response $\{a_{k}\}_{k=0}^{L-\nu-1}$ of length $L-\nu$, and an open-loop measured trajectory $\{u_{k},y_{k}\}_{k=0}^{N-1}$ with persistently exciting input signal of the system $G$ is available. Based on these ingredients, necessary and sufficient conditions for closed-loop $L$-dissipativity for the channels $r \mapsto z$, $r \mapsto e$ and $r \mapsto u$ are presented. Before stating our main result, we provide a characterization of all closed-loop trajectories with zero initial conditions. To this end, we need to ensure that the feedback loop is well-posed.
\begin{definition}[{\cite{daleh2004}}]
	The standard feedback loop (\reffig{closed_loop}) is well-posed if all signals $e$, $u$ and $z$ in the feedback loop are uniquely defined for every choice of the system state variables for both the controller and the plant and every choice of the external input r.
	\label{def:well-posedness}
\end{definition}
\refdef{well-posedness} translates into the following Proposition in the standard feedback loop.
\begin{proposition}[{\cite{daleh2004}}]
	Suppose $(A,B,C,D)$ is a realization of the SISO LTI plant $G$ and $(A_{\mathrm{c}},B_{\mathrm{c}},C_{\mathrm{c}},D_{\mathrm{c}})$ is a realization of the SISO LTI  controller $K$. Then, the feedback loop is well-posed if and only if $1+DD_{\mathrm{c}} \neq 0$.
	\label{prop:well-posedness}
\end{proposition}
Furthermore, well-posedness also guarantees the existence of a well-defined closed-loop state-space realization \cite{daleh2004}. Well-posedness can indeed be checked in our data-driven framework via \refprop{well-posedness}, since the feedthrough term of the plant can be calculated from the given data. To this purpose, one can use \reftheo{data_based_trajectory_representation} and $V_{L}^{\nu}(u,y)$ with $\nu \geq n$ to construct an input-output trajectory with zero initial conditions of $G$ and simply divide the output at the first time step by the input. The following result provides a data-based characterization of all closed-loop trajectories.
\begin{proposition}
	Suppose the standard feedback loop in \reffig{closed_loop} is well-posed and $\{u_{k},y_{k}\}_{k=0}^{N-1}$ is a trajectory of $G$, where u is persistently exciting of order $L+n$. Then, for any $\nu$ with $n \leq \nu < L$, $\{r_{k},z_{k}\}_{k=0}^{L-\nu-1}$ is a closed-loop trajectory of the standard feedback loop with zero initial conditions if and only if there exists a vector $\beta \in \mathbb{R}^{\mathrm{dim}(\mathrm{col}(V_{L}^{\nu}(u,y)))}$ such that
	\begin{equation}
		M_{L-\nu}(a) J H_{L}(u,y) V_{L}^{\nu}(u,y) \beta = \begin{pmatrix} r \\ z \end{pmatrix},
		\label{eq:closed_loop_trajectories}
	\end{equation}
where $M_{L-\nu}(a) = \begin{pmatrix}I_{L-\nu} & T_{L-\nu}(a) \\ 0_{(L-\nu) \times (L-\nu)} & T_{L-\nu}(a)\end{pmatrix}$, $J = \begin{pmatrix} J_{L}^{\nu} & 0_{(L-\nu) \times L} \\ 0_{(L-\nu) \times L} & J_{L}^{\nu}\end{pmatrix}$ with \\ $J_{L}^{\nu} = \begin{pmatrix} 0_{(L-\nu) \times \nu} & I_{L-\nu} \end{pmatrix}$ and $\{a_{k}\}_{k=0}^{L-\nu-1}$ is the impulse response of $K$.
\label{prop:data_based_closed_loop_trajectories}
\end{proposition}

\begin{proof}
	\textbf{if:} For a fixed $\beta \in \mathbb{R}^{\mathrm{dim}(\mathrm{col}(V_{L}^{\nu}(u,y)))}$
\begin{equation}
	H_{L}(u,y) V_{L}^{\nu}(u,y)\beta = \begin{pmatrix} \bar{u} \\ \bar{y} \end{pmatrix}
	\label{eq:cl_trajectory_representation_proof_1}
\end{equation}
is a trajectory $(\bar{u}$, $\bar{y})$ of length $L$ of $G$ with $\bar{u}_{0} = \hdots = \bar{u}_{\nu-1} = 0$ and $\bar{y}_{0} = \hdots = \bar{y}_{\nu-1} = 0$, by using \reftheo{data_based_trajectory_representation} combined with the definition of $V_{L}^{\nu}(u,y)$. Since $\nu \geq n$ by assumption, $(\bar{u},\bar{y})$ is a trajectory of $G$ with $\bar{x}_{k}=0$ for $k=0,\hdots,\nu-1$, where $\bar{x}$ is the corresponding state in some minimal realization. Multiplying \refeq{cl_trajectory_representation_proof_1} with $J$ from the left yields
\begin{equation}
	J H_{L}(u,y) V_{L}^{\nu}(u,y)\beta = \begin{pmatrix} J_{L}^{\nu} & 0_{(L-\nu) \times (L)} \\ 0_{(L-\nu) \times (L)} & J_{L}^{\nu}\end{pmatrix} \begin{pmatrix} \bar{u} \\ \bar{y} \end{pmatrix} = \begin{pmatrix} \hat{u} \\ \hat{y} \end{pmatrix},
	\label{eq:closed_loop_trajectory_parametrization_shortened}
\end{equation}
where $(\hat{u},\hat{y})$ contains the last $L-\nu$ entries of $(\bar{u},\bar{y})$ and is therefore a trajectory of length $L-\nu$ of $G$ with zero initial conditions. Representing the controller $K$ via the Toeplitz matrix $T_{L-\nu}(a)$ implies also zero initial conditions by assumption for the controller. Therefore, using the commutativity of SISO LTI systems with zero initial conditions, the standard feedback loop has the same input-output behavior from $r \mapsto z$ as the transformed loop shown in \reffig{closed_loop_interchanged}. Multiplying \refeq{closed_loop_trajectory_parametrization_shortened} by $M_{L-\nu}(a)$ from the left, we obtain $M_{L-\nu}(a) J H_{L}(u,y) V_{L}^{\nu}(u,y) \beta = M_{L-\nu}(a) \begin{pmatrix} \hat{u}^{\top} & \hat{y}^{\top}  \end{pmatrix}^{\top}$. Then, using the signal definition of the interchanged standard feedback loop in \reffig{closed_loop_interchanged} leads to
\begin{equation}
		M_{L-\nu}(a) \begin{pmatrix} \hat{u} \\ \hat{y}  \end{pmatrix} = M_{L-\nu}(a) \begin{pmatrix} \hat{e} \\ \hat{w} \end{pmatrix}=\begin{pmatrix}\hat{e} + T_{L-\nu}(a) \hat{w} \\ T_{L-\nu}(a) \hat{w}\end{pmatrix} =\begin{pmatrix}\hat{e} + \hat{z} \\ \hat{z}\end{pmatrix} =\begin{pmatrix}\hat{r} \\ \hat{z}\end{pmatrix},
	\label{eq:signal_transformation}
\end{equation}
where $\hat{e}=\hat{u}$, $\hat{w}= \hat{y}$ and $\hat{z}=T_{L-\nu}(a) \hat{w}$ (see \reffig{closed_loop_interchanged}). 
Since the closed loop is well-posed, we can construct a state-space realization of the resulting closed loop by stacking the states of an arbitrary realization of the controller and the states of an arbitrary realization of the plant. As we assumed zero initial conditions for both, the controller and the plant, also the constructed closed-loop realization has zero initial conditions. Therefore, $(\hat{r},\hat{z})$ is a trajectory of length $L-\nu$ of the closed loop with zero initial conditions, which, together with \refeq{signal_transformation}, concludes the "if"-part.

\textbf{Only if:} Suppose $\{\hat{r}_{k},\hat{z}_{k}\}_{k=0}^{L-\nu-1}$ is a closed-loop trajectory of the standard feedback loop (\reffig{closed_loop}) with zero initial conditions. Since a closed-loop state-space realization can be constructed by stacking the individual states of realizations of $K$ and $G$, zero initial conditions for the closed loop imply zero initial conditions for both the controller and the plant. Similar to the if-part, we use the commutativity property and reverse the steps seen in \refeq{signal_transformation} to guarantee the existence of a trajectory $(\hat{u},\hat{y})$ of length $L-\nu$ of $G$ with zero initial conditions, which satisfies \refeq{signal_transformation}. Hence, we can artificially insert zeros to deduce that $\begin{pmatrix}\bar{u}^{\top} & \bar{y}^{\top}\end{pmatrix} = \begin{pmatrix} 0_{1 \times \nu} & \hat{u}^{\top} & 0_{1 \times \nu} & \hat{y}^{\top} \end{pmatrix}$ is a trajectory of length $L$ of $G$. Thus, using \reftheo{data_based_trajectory_representation} and the definition of $V_{L}^{\nu}(u,y)$, there exists a vector $\beta$ satisfying \refeq{closed_loop_trajectories}.
\end{proof}
\begin{remark}
Similar to \refprop{data_based_closed_loop_trajectories}, it is possible to parametrize closed-loop input-output trajectories corresponding to further channels, for instance $r \mapsto e$ (reference to error) or $r \mapsto u$ (reference to control variable). To this purpose, one has to change the matrix $M_{L-\nu}(a)$ in \refeq{closed_loop_trajectories}. For the channels $r \mapsto e$ and $r \mapsto u$, \refprop{data_based_closed_loop_trajectories} holds with
\begin{equation}
	M_{L-\nu}(a) = \begin{pmatrix}I_{L-\nu} & T_{L-\nu}(a) \\ I_{L-\nu} & 0_{(L-\nu) \times (L-\nu)} \end{pmatrix} \quad \text{and} \quad M_{L-\nu}(a) = \begin{pmatrix}I_{L-\nu} & T_{L-\nu}(a) \\ T_{L-\nu}(a) & 0_{(L-\nu) \times (L-\nu)} \end{pmatrix},
\end{equation}
respectively. The proof for these channels goes along the same lines as for $r \mapsto z$.
	\label{rema:feedback_channels}
\end{remark}
\cite{markovsky2010} provides an approach for closed-loop data-driven simulation for a given controller, allowing to compute the closed-loop response for a given input reference. On the other hand \refprop{data_based_closed_loop_trajectories} allows us to parametrize all input-output trajectories (with zero initial conditions) jointly, leading to a natural extension of the open-loop parametrization in \reftheo{data_based_trajectory_representation}. More importantly, the presented parametrization is linear in the controller parameters which is essential for using this result for data-driven controller design. By using our representation in \refprop{data_based_closed_loop_trajectories} of all closed-loop trajectories with zero initial conditions, we can state the following result for verifying closed-loop dissipativity of an unknown plant $G$ with a given controller $K$.
\begin{theorem}
	Suppose the standard feedback loop (\reffig{closed_loop}) is well-posed and $\{u_{k},y_{k}\}_{k=0}^{N-1}$ is a trajectory of $G$, where u is persistently exciting of order $L+n$. Then, for any $\nu$ with $n \leq \nu < L$, the channel $r \mapsto z$ is $(L-v)$-dissipative w.r.t. the supply rate $\Pi$ if and only if
	\begin{equation}
			V_{L}^{\nu}(u,y)^{\top} H_{L}(u,y)^{\top} J^{\top} M_{L-\nu}(a)^{\top} \Pi_{L-\nu} M_{L-\nu}(a) J H_{L}(u,y)  V_{L}^{\nu}(u,y) \succcurlyeq 0,
	\label{eq:closed_loop_dissipative_cond}
	\end{equation}
	where $M_{L-\nu}(a)$ as defined in \refprop{data_based_closed_loop_trajectories}.
	\label{theo:data_based_cl_L-dissipativity}
\end{theorem}
\begin{proof}
By using the dissipativity condition \refeq{L-dissipativity_cond} for the horizon $L-\nu$ we get
\begin{equation*}
	\sum_{k=0}^{L-\nu-1} \begin{pmatrix} r_{k} \\ z_{k} \end{pmatrix}^{\top} \Pi \begin{pmatrix} r_{k} \\ z_{k} \end{pmatrix} \geq 0 \quad \text{if and only if} \quad \begin{pmatrix} r \\ z \end{pmatrix}^{\top} \Pi_{L-\nu} \begin{pmatrix} r \\ z \end{pmatrix} \geq 0
\end{equation*}
for all trajectories $(r,z)$ of length $L-\nu$ with zero initial conditions. Furthermore, using \refprop{data_based_closed_loop_trajectories}, this turns out to be equivalent to \refeq{closed_loop_dissipative_cond}.
\end{proof}
\reftheo{data_based_cl_L-dissipativity} provides a validation technique for closed-loop $(L-\nu)$-dissipativity in the standard feedback loop. Only one open-loop measured trajectory of the plant and the model of the controller is needed to verify dissipativity. The condition is very simple in the sense that only the definiteness of a single matrix has to be checked, which can be easily done numerically. Similar to \reftheo{data_based_cl_L-dissipativity} , it is possible to obtain closed-loop dissipativity conditions for other channels by using the appropriate matrix $M_{L-\nu}(a)$ as discussed in \refrema{feedback_channels}. The main advantages of the proposed method compared to applying \reftheo{data_based_L-dissipativity} to closed-loop data is that the controller does not have to be implemented and no new measurements have to be taken. More importantly, the result in \reftheo{data_based_cl_L-dissipativity} can be used for controller design as shown in the next section.

\section{Data-driven controller synthesis for closed-loop dissipativity} \label{sec:data_driven_controller_synthesis_for_closed_loop_dissipativity}
In this section, we present a method for controller design with guaranteed closed-loop $L$-dissipativity based on the data-driven controller validation for $L$-dissipativity in \reftheo{data_based_cl_L-dissipativity}. More precisely, given a single measured trajectory ${\{{u}_{k},{y}_{k}\}}^{N-1}_{k=0}$ of the plant $G$, the goal is to find a controller $K$ such that the closed loop is $(L-\nu)$-dissipative w.r.t. a desired supply rate $\Pi$. We look for a suitable finite impulse response ${\{{a}_{k}\}}^{L-\nu-1}_{k=0}$ of the controller $K$, that satisfies a desired dissipativity specification
\begin{equation}
	V_{L}^{\nu}(u,y)^{\top} H_{L}(u,y)^{\top} J{^\top} M_{L-\nu}(a)^{\top} \Pi_{L-\nu} M_{L-\nu}(a) J H_{L}(u,y)  V_{L}^{\nu}(u,y) \succcurlyeq 0.
	\label{eq:synthesis_inequality}
\end{equation}
Instead of solving \refeq{synthesis_inequality} directly for an impulse response $a$, we rather parametrize the impulse response by some optimization variable $p \in \mathbb{R}^{d}$. More precisely, we parametrize the Toeplitz matrix $T_{L-\nu}(a(p)) = \sum_{i=0}^{d-1} p_{i} T_{L-\nu,i}$ by a linear combination of lower triangular basis matrices $T_{L-\nu,i}$. In this way, we can include a desired structure and impose causality in our controller in the design process. For example, choosing the parameter vector $p=\begin{pmatrix} K_p & K_i \end{pmatrix}^{\top}$ and basis matrices $T_{L-\nu,0} = I_{L-\nu}$ and $T_{L-\nu,1} =\begin{psmallmatrix}0 & 0 & \cdots & 0 \\ T_{s} & 0 & \cdots & 0 \\ \vdots & \ddots & \ddots & \vdots \\ T_{s} & \cdots & T_{s} & 0\end{psmallmatrix}$, where $T_{s}$ is the sampling rate, represents a discretized PI controller. Such a parametrization retains the possibility to solve \refeq{synthesis_inequality} for the whole impulse response $a$, by simply choosing $p=a$ and $T_{L-\nu,i}$ as square matrices of size $L{-}\nu$ with ones on the $i$-th diagonal below the main diagonal, i.e., $i=0$ represents the main diagonal. If $T_{L-\nu}(a(p))$ is linear in p then $M_{L-\nu}(a(p))$ is linear in $p$ for all channels in the standard feedback loop (see \refrema{feedback_channels}), i.e., inequality \refeq{synthesis_inequality} results in a QMI in $p$, which can be expressed in the standard form
\begin{equation}
	Q(p) = Q_{0} + \sum_{i=0}^{d-1} p_{i} Q_{i} + \sum_{i=0}^{d-1}\sum_{j=0}^{d-1} p_{i} p_{j} Q_{ij} \preccurlyeq 0,
	\label{eq:qmi}
\end{equation}
where $Q_{0}$, $Q_{i} \in \mathbb{S}$ for $i=0 \hdots d-1$ and $Q_{ij} \in \mathbb{S}$ for $i,j=0 \hdots d-1$ are symmetric matrices of the same dimension. In the following, we only state the matrices $Q_{0}$, $Q_{i}$ and $Q_{ij}$ for the channel $r \mapsto z$, while other channels can be considered in a similar fashion. To do this, we use the abbreviation $N_{u}(u,y){=} J_{L}^{\nu}  {H_{L}(u)} V_{L}^{\nu}(u,y)$ and $N_{y}(u,y){=} J_{L}^{\nu} {H_{L}(y)} V_{L}^{\nu}(u,y)$. The synthesis QMI matrices are
\begin{equation*}
\begin{aligned}
	Q_{i} &= -N_{u}^{\top}(u,y) (Q_{L-\nu}+S_{L-\nu}) T_{L-\nu,i} N_{y}(u,y) - N_{y}^{\top}(u,y) T_{L-\nu,i}^{\top} (Q_{L-\nu}+S_{L-\nu}^{\top}) N_{u}(u,y), \\
	V_{ij} &= - N_{y}^{\top}(u,y) T_{L-\nu,i}^{\top} (Q_{L-\nu}+S_{L-\nu}+S_{L-\nu}^{\top}+R_{L-\nu}) T_{L-\nu,j} N_{y}(u,y), \\
	Q_{ij} &= \left\{ \begin{array}{ll} V_{ij}, & \text{if } i=j \\ V_{ij}+V_{ji}, & \text{if } i < j \\ 0,& {\text{else}}\end{array}\right., \quad Q_{0} = -N_{u}^{\top}(u,y) Q_{L-\nu} N_{u}(u,y).
\end{aligned}
\end{equation*}
We employ a difference of convex functions programming approach to find a feasible solution to the QMI \refeq{synthesis_inequality} with the new controller parametrization $p$. The DC programming approach for finding a feasible solution to a general nonconvex QMI is described in \cite{niu2014}. We consider the algorithm proposed in \cite{niu2014} to find a feasible controller parametrization $p$ for the desired finite-horizon dissipativity specification. The algorithm solves a linear semidefinite program (SDP) in each iteration to converge to a local solution of the QMI. For more details on the theory about the QMI feasibility problems and how to solve them via DC progamming we refer to \cite{niu2014} and the references therein.

For controller design, it is often desirably to include several dissipativity specifications on one or more channels, i.e., to pursue a multi-objective formulation, which is studied, for example, in \cite{scherer1997} assuming that a model is available. In the developed setup in this paper, several different dissipativity specifications on different channels can easily be included by diagonal augmentation of the corresponding synthesis QMIs to one larger QMI. 
\begin{algorithm}[htbp]
\floatconts
  {algo:synthesis}%
  {\caption{Controller design procedure.}}%
{%
	\KwResult{Controller parameterization $p$ that satisfies finite-horizon dissipativity specification.}
	\begin{enumerate}
	\item Choose performance specifications in the form of one or several supply rates $\Pi$ for each channel and decide on a controller parametrization $T_{L-\nu,i}$.
	\item Collect data $\{u_{k},y_{k}\}_{k=0}^{N-1}$ of the plant $G$ with $u$ persistently exciting of order $L+\nu$.
    	\item Apply {\cite[][DC Algorithm 1]{niu2014}} to find a feasible solution $p$ to the corresponding synthesis inequality \refeq{qmi}.
	\item Apply the controller $u = T_{L-\nu}(a(p)) e$.
    \end{enumerate}
}%
\end{algorithm}
In \refalgo{synthesis} we describe the general controller design procedure for finite-horizon dissipativity. \refalgo{synthesis} allows us to design a controller for an unknown system achieving closed-loop $(L{-}\nu)$-dissipativity, where $\nu$ is an upper bound on the system order of the plant $G$, using only measured data. The solution to the above design problem is in general non-unique, which is due to the fact that multiple controllers with the imposed linear controller structure can satisfy the dissipativity specification. The DC algorithm from step 3 in \refalgo{synthesis} returns a local solution in the neighborhood of its initial value. By considering multiple channels in the design specification, the presented approach solves a multi-objective control problem. Further, since we can impose desired structure on the controller matrices $T_{L-\nu,i}$, the presented approach can also be seen as a structured controller design problem, which is generally difficult in model-based multi-objective control via LMI techniques (\cite{scherer1997}).
\begin{remark}
	\refalgo{synthesis} allows us to perform mixed sensitivity design (see \cite{kwakernaak2016}), where the transfer functions from $r \mapsto e$ and $r \mapsto u$ are shaped trough the introduction of additional filters. To this purpose, filters $W_{e}$ and $W_{u}$ for both channels with corresponding impulse responses $w_{e}$ and $w_{u}$ are designed. These filters are added artificially to the outputs of the channels and shape the sensitivity transfer functions when replacing \refeq{qmi}, e.g., by 
\begin{equation}
	\begin{pmatrix} *\end{pmatrix}^{\top} \begin{pmatrix}I & 0 \\ 0 & -I\end{pmatrix} \begin{pmatrix}I & 0 \\ 0 & T_{L-\nu}(w_{j})\end{pmatrix} M_{L-\nu}(a(p)) \begin{pmatrix} N_{u}(u,y) \\ N_{y}(u,y)\end{pmatrix} {\succcurlyeq} 0,
\end{equation}
for an $L_{2}$-gain dissipativity specification for $j=e,u$, from the the input of the channels to the output of the filters. The transfer funtions are then shaped in the sense that (for a sufficiently long horizon $L-\nu$) the magnitude response of the inverses of these filters is an upper bound for the magnitude response of the sensitivity transfer functions.
	\label{rema:mixed_sensitivity}
\end{remark}
\begin{remark}
Additional linear matrix inequality (LMI) or QMI conditions on the controller parameters $p$ can also be added. For example, one can bound the matrix norm of the Toeplitz matrix $T_{L-\nu}(a(p))$ to bound the control energy or even achieve finite horizon input-output stability by using similar arguments as in the small gain criterion. More precisely, one can calculate the finite horizon $L_{2}$-gain $\gamma$ of the plant $G$ via \reftheo{data_based_L-dissipativity} and impose a strict upper bound of $\frac{1}{\gamma}$ to the norm of $T_{L-\nu}(a(p))$ to obtain a finite-horizon $L_{2}$ gain $\gamma_{\mathrm{cl}} < 1$ of the closed loop. The condition on the controller can be implemented by adding the convex constraint $T_{L-\nu}(a(p))^{\top}T_{L-\nu}(a(p)) \prec \frac{1}{\gamma} I_{L-\nu}$ to the QMI feasibility problem. For a large enough horizon $L-\nu$, this implies that the (infinite-horizon) closed-loop $L_{2}$-gain is bounded by $\frac{1}{\gamma}$, compare \cite{koch2020}.
\label{rema:finite_horizon_stability}
\end{remark}
Note that if the resulting QMI \refeq{qmi} for the controller design is convex, it can be transformed into an LMI via a Schur complement and solved via LMI optimization techniques. A sufficient condition for convexity of the QMI \refeq{synthesis_inequality} for the channel $r \mapsto z$ is $Q_{L-\nu}+S_{L-\nu}+S_{L-\nu}^{\top}+R_{L-\nu} \preccurlyeq 0$. This condition can be restrictive but it holds, e.g., for an $L_{2}$-gain bound $\gamma \leq 1$ with the matrices $Q = 1$, $S=0$ and $R=-1$. Note that these conditions vary with the choice of the channel in the feedback loop. Nevertheless, if it is possible, translating the QMI into an LMI is beneficial since LMIs can be solved more efficiently and a global solution can be found. Throughout this paper, we assumed noise-free measurements, which can be restrictive in practice. A simple remedy, proposed in \cite{romer2019}, for handling noisy data is to use the relaxation
\begin{equation*}
	V_{L}^{\nu}(u,y)^{\top} H_{L}(u,y)^{\top} J{^\top} M_{L-\nu}(a(p))^{\top} \Pi_{L-\nu} M_{L-\nu}(a(p)) J H_{L}(u,y)  V_{L}^{\nu}(u,y) \succcurlyeq \delta I, \quad \text {with } \delta < 0.
\end{equation*}

\section{Numerical example} \label{sec:numerical_example}
In this section, we apply our method to control a two dimensional ($n{=}2$) linearized and discretized two tank system. The control input is the inflow to the first tank and the output of the system is the level in the second tank, which we want to control. The linearized system is discretized with a sampling rate $T_{s} {=} 0.5\mathrm{s}$, which leads to the plant $G$
\begin{equation*}
	x_{k+1} = \begin{pmatrix}0.9677 & 0 \\ 0.0317 & 0.9677\end{pmatrix} x_{k} + \begin{pmatrix}0.1363 \\ 0.0022\end{pmatrix} u_{k}, \quad y_{k}= \begin{pmatrix}0 & 1\end{pmatrix} x_{k}.
\end{equation*}
We use the specification of the proposed mixed sensitivity design in \refrema{mixed_sensitivity} with $L{=}110$ and $\nu{=}2 \geq n$ to guarantee dissipativity over the horizon $L{-}\nu {=} 108$. Next, we choose a PI parametrization $p{=}\begin{pmatrix}K_p & K_i\end{pmatrix}^{\top}$ for the controller. To apply \refalgo{synthesis}, we first generate a random persistently exciting input $u$ of order $L{+}n$ with a length of $N{=}223$ unifomly from $(-10,10)$ and measure the corresponding output $y$ of $G$. Second, we find the feasible controller parametrization $K_p {=}0.1551$ and $K_i {=} 0.0084$ to the resulting synthesis QMI \refeq{qmi} from the dissipativity specifications via DC programming by implementing the algorithm from \cite{niu2014} in YALMIP (\cite{lofberg2004}). Finally, we apply the input $u {=} T_{L-\nu}(a(p)) e$ in the feedback loop. To analyze the closed-loop behavior without use of the plant model, one can use \refprop{data_based_closed_loop_trajectories} to compute significant input-output trajectories, e.g., the step response for the channel $r \mapsto z$ (see \reffig{step_response_cl}).
\begin{figure}[h]
	\centering
	\begin{minipage}{0.48\textwidth}
\centering
	\includegraphics[width=\textwidth]{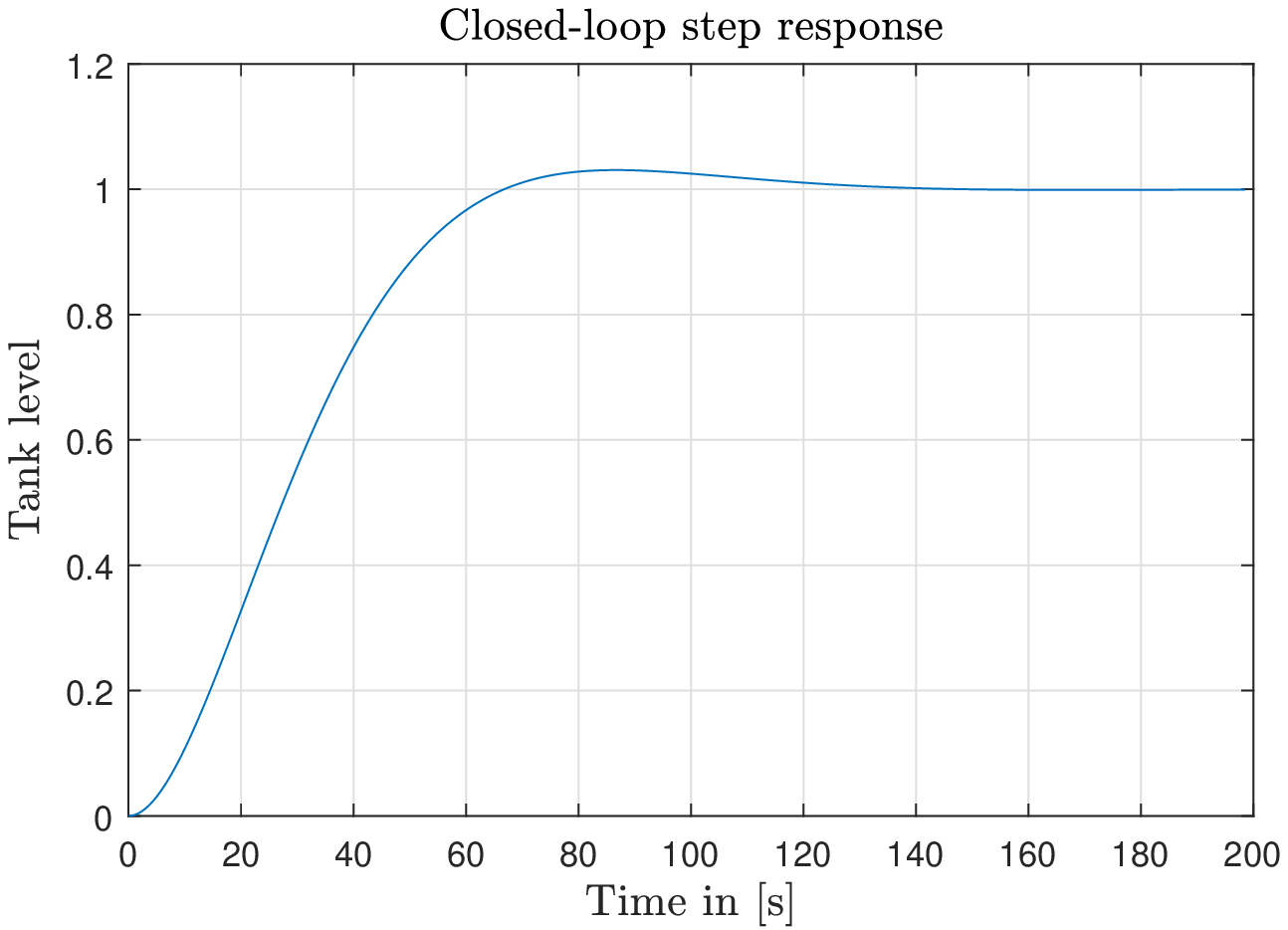}
	\caption{Step response from $r \mapsto z$.}	
	\label{fig:step_response_cl}	
\end{minipage}
\begin{minipage}{0.48\textwidth}
	\centering
	\includegraphics[width=\textwidth]{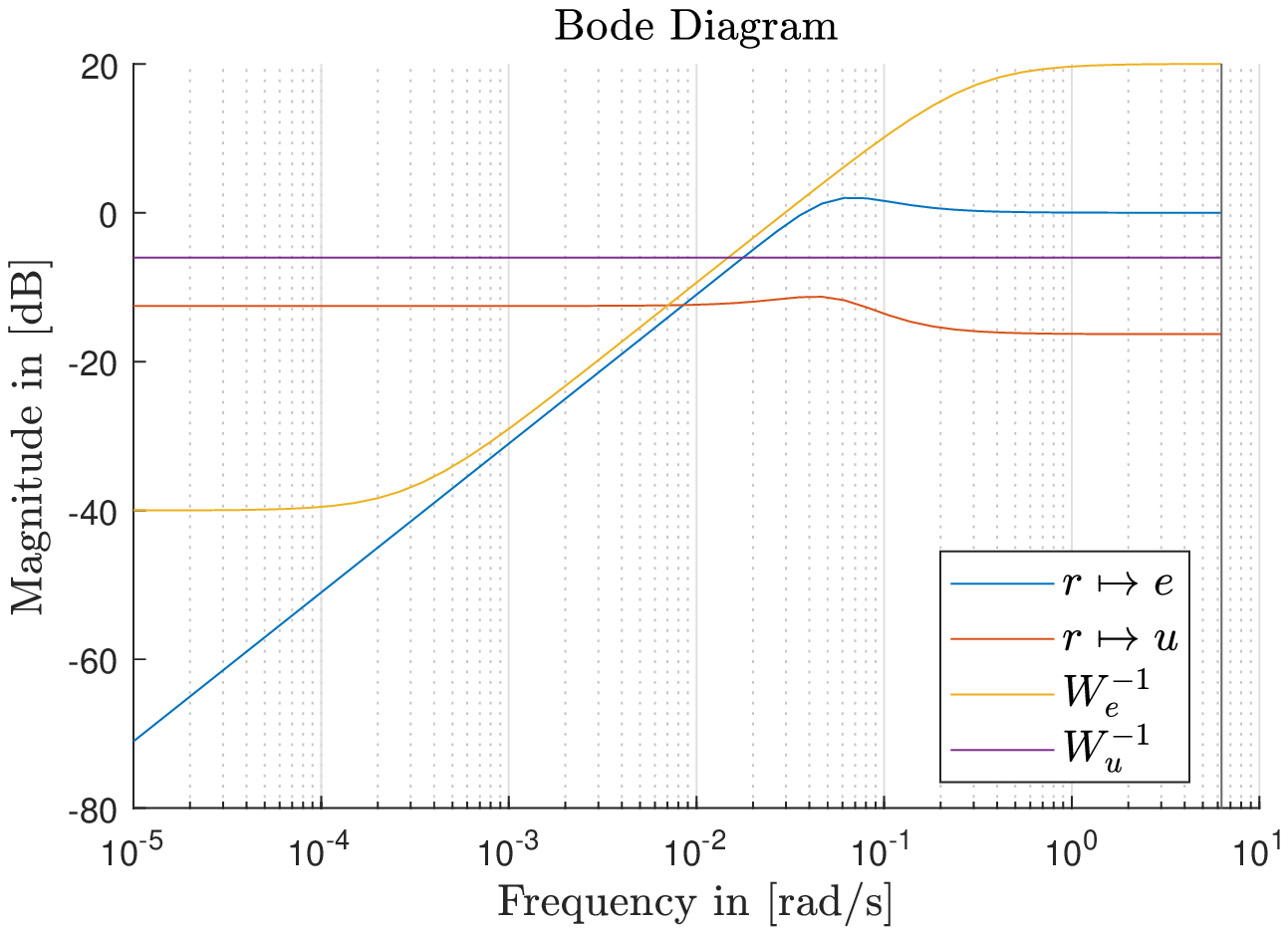}
	\caption{Magnitude responses.}
	\label{fig:magnitude_response_cl}
\end{minipage}
\end{figure}
We see that the reference is tracked without steady state error in a smooth manner. \reffig{magnitude_response_cl} shows the closed-loop magnitude responses of the sensitivity transfer functions and the satisfaction of the imposed constraints from the filters in the mixed sensitivity design. Note that plotting this magnitude plot requires knowledge of the plant model. In conclusion, we designed a PI controller for reference tracking directly from data while imposing finite-horizon mixed-sensitivity dissipativity constraints.

\section{Conclusion} \label{sec:conclusion}
We presented a data-driven framework for verifying closed-loop $L$-dissipativity based on a given controller and one input-output trajectory of the plant in the standard feedback loop. These results can then be employed to design controllers on the basis of one input-output trajectory guaranteeing desired dissipativity properties on different channels of the closed loop. The controller synthesis approach requires the solution of a QMI, which can be obtained from a DC programming approach. It is an interesting issue for future research to extend the presented results to multi-input multi-output systems, which can in general not be handled by the framework since they are not commutative in most cases. In addition, since data is usually corrupted by noise and real systems are not deterministic and LTI, it is important to extend the theoretical guarantees to such cases.

\acks{This work was funded by Deutsche Forschungsgemeinschaft (DFG, German Research Foundation) under Germany’s Excellence Strategy - EXC 2075 - 390740016. The authors thank the International Max Planck Research School for Intelligent Systems (IMPRS-IS) for supporting Julian Berberich and Anne Koch.}

\bibliography{references}

\end{document}